%% file: ldsic_techreport.tex
\documentclass[12pt]{article}

\usepackage{url}
\usepackage{amsthm}
\usepackage{graphicx}
\usepackage{multicol}
\usepackage{xspace}
\usepackage{mathrsfs}
\usepackage{amsmath}
\usepackage{amssymb}
\usepackage{mathrsfs}
\usepackage{graphics}
\usepackage{enumerate}
\usepackage{booktabs}
\title{%
\allowbreak Sparse Polynomial Interpolation Codes and their decoding beyond
half the minimal distance\raisebox{0ex}{*}%
}%title

\def\addauthorsection{\relax}

\author{
Erich L. Kaltofen\\
Dept.\ of Mathematics, NCSU\\
Raleigh, NC 27695, USA\\
  \url{kaltofen@math.ncsu.edu}\\
  \url{www4.ncsu.edu/~kaltofen},\\
Cl\'{e}ment Pernet\\
  {U. J. Fourier, LIP-AriC, CNRS, Inria, UCBL, \'ENS de Lyon}\\
  {46 All\'ee d'Italie, 69364 Lyon Cedex 7, France}\\
   \url{clement.pernet@imag.fr}\\
   \url{http://membres-liglab.imag.fr/pernet/}
}
\newtheorem{theorem}{Theorem}%mtc3 [section]
\newtheorem{definition}{Definition}%mtc3 [section]
{Proposition}
\newtheorem{cor}%mtc3 [theorem]%
{Corollary}
\newtheorem{lemma}%mtc3 [theorem]%
{Lemma}
\newtheorem{remark}{Remark}%mtc3 [section]
\newtheorem{example}{Example}%mtc3 [section]
\newtheorem{problem}{Problem}%mtc3 [section]
\newcommand{\graphl}{.8\textwidth}
\newcommand{\BM}{Ber\-le\-kamp/\allowbreak Mas\-sey\xspace}
\newcommand{\BMA}{\BM algorithm\xspace}
\newcommand{\ZZ}{{\mathbb Z}}
\newcommand{\RR}{{\mathbb R}}

\newcommand{\KK}{\mathsf{K}}
\newcommand{\QQ}{{\mathbb Q}}
\newcommand{\seq}[1]{(#1)}
\newcommand{\gendeg}{t}

\newcommand{\sparpoly}{f} % sparse polynomial
\newcommand{\genf}{\Lambda} % generator
\newcommand{\spbound}{T}
\newcommand{\GO}[1]{\ensuremath{{O}(#1)\xspace}}
\newcommand{\Z}{\ensuremath{\mathbb{Z}}}
\newcommand{\comment}[1]{}
\newcommand{\EKhref}[2]{URL: \url{#1}}
\usepackage{breakurl}
\usepackage[%
breaklinks%
,colorlinks%
,linkcolor=red% internal document links
,anchorcolor=blue%
,citecolor=blue%
,bookmarks=false%
]{hyperref}

\def\citep{\cite}

\begin{document}

\input{ldsic}

\def\refname{\Large\bfseries References}

%mtc3 \bibliographystyle{myplainnat}
\bibliographystyle{acm} %ek5 {abbrv}
%EK make references smaller
%% \makeatletter
%% \def\thebibliography#1{%
%% %\ifnum\addauflag=0\addauthorsection\global\addauflag=1\fi
%%      \section[References]{%    <=== OPTIONAL ARGUMENT ADDED HERE
%%         {References} % was uppercased but this affects pdf bookmarks (SP/GM October 2004)
%%           \vskip -9pt  % GM July 2000 (for tighter spacing)
%%          \@mkboth{{\refname}}{{\refname}}%
%%      }%
%%      \list{[\arabic{enumi}]}{%
%%          \settowidth\labelwidth{[#1]}%
%%          \leftmargin\labelwidth
%%          \advance\leftmargin\labelsep
%%          \advance\leftmargin\bibindent
%%          \parsep=0pt\itemsep=1pt % GM July 2000
%%          \itemindent -\bibindent
%%          \listparindent \itemindent
%%          \usecounter{enumi}
%%      }%
%%      \let\newblock\@empty
%%      \raggedright % GM July 2000
%%      \sloppy
%%      \sfcode`\.=1000\relax
%%      %EK make references small
%%      \vspace*{0.0mm}
%%      \scriptsize
%% }%thebibliography
%% \makeatother

\bibliography{strings,ldsic,crossrefs}

\end{document}

%% file: ldsic.tex
\maketitle %ek5 move within vim fold, remove extra \gt\gt\gt
%\urlstyle{rm}%ek5 for references
%>>>

\def\thefootnote{\fnsymbol{footnote}}
\footnotetext[1]{%<<<
\scriptsize%ek5 was: \footnotesize
This material is based on work supported in part
by the National Science Foundation under Grant %EK3 remove CCF-0830347 and
CCF-1115772 (Kaltofen),
%ELK3 rewrite and the French ANR under Grant HPAC ANR-11-BS02-013 (Pernet), and the Inria associate team grant QOLAPS (Pernet).
and the Agence Nationale de la Recherche %ELK3 to match spelling out NSF
under Grant HPAC ANR-11-BS02-013 and the Inria Associate Teams %ELK3 note the plural, caps
Grant QOLAPS (Pernet).
%mtc3 
%mtc3 \par
%Clement: your support
}%footnotetext %>>>

\begin{abstract}\noindent %<<<
\input{abstract}
\end{abstract} %>>>

%mtc3 \vspace{1mm} \noindent {\bfseries Categories and Subject Descriptors:} %<<<
%mtc3 I.1.2 
%mtc3 {[Symbolic and Algebraic Manipulation]}: Algorithms;
%mtc3 G.1.6 {[Numerical Analysis]}: Global optimization 
%mtc3 
%mtc3 \vspace{1mm} \noindent {\bfseries General Terms:}
%mtc3 algorithms, experimentation
%mtc3 
%mtc3 \vspace{1mm} \noindent {\bfseries Keywords:}
%mtc3 semidefinite programming, 
%mtc3 sum-of-squares, 
%mtc3 Hilbert-Artin representation,
%mtc3 %denominator, 
%mtc3 certificates of infeasibility, 
%mtc3 Semidefinite Farkas' lemma %>>>

\vspace{1mm} \noindent {\bfseries Categories and Subject Descriptors:}
%Third-level node number {[Second-level node title]}: 
%Third-level node title---first subject descriptor, second subject descr...

\noindent
I.1.2 {[Symbolic and Algebraic Manipulation]}: Algorithms%
; G.1.1 {[Numerical Analysis]}: Interpolation--smoothing%
%{\color{red}
%; G.1.3 {[Numerical Analysis]}: Numerical Linear Algebra---error analysis%
; 
\\E.4 {[Coding and Information Theory]}: Error control codes.%
%}

\vspace{1mm} \noindent {\bfseries General Terms:}
  Algorithms%
, Reliability

\vspace{1mm} \noindent {\bfseries Keywords:}
sparse polynomial interpolation,
Blahut's algorithm,
Prony's algorithm,
exact polynomial fitting with errors.

\newcommand{\myparagraph}[1]{{\bfseries\itshape #1.}}

\input{1intro}

\section{Sparse interpolation codes}
\label{sec:sic}

\begin{definition}
Let $\KK$ be a field, $0<n\leq m$, two integers and 
let %ELK3 add
$x_0,\dots,x_{n-1}$ be %ELK3 was: ,
$n$ distinct elements of\/ $\KK$. %ELK3 \/ is the italics adjustment
    A sparse polynomial evaluation code of parameters $(n,\spbound)$ over~$\KK$
    is defined as the set 
\[
\begin{split}
\mathcal{C}(n,\spbound)=\left\{ (f(x_0), f(x_1),\dots,f(x_{n-1})) :
    f\in \KK[z]\right.\\ \left.  \text{ is } t\text{-sparse \allowbreak
with } t \leq \spbound \text{ and } \deg f<m\right\}
\end{split}
\]
\end{definition}

In order to benefit from Blahut/Ben-Or/Tiwari  %ELK3 was: Ben-Or \& Tiwari's
%CP2: reverted to Ben-Or \& Tiwari : there is no known Blahut algorithm and I
%prefer to avoid naming an existing algorithm
algorithm for error free
interpolation, we will consider, until section~\ref{sec:charzero}, the special
case where the evaluation points are 
consecutive powers of a primitive $m$-th root of unity $\alpha \in \KK$:
$x_i=\alpha^i$. In this context, we can state the minimum distance of such codes
provided that $2\spbound$ divides $m$.

\begin{theorem}\label{th:mindist}
If $\alpha\in \KK$ is a primitive $m$-th root of unity, $x_i=\alpha^i, \
i\in\{0,\dots,n-1\}$ and 
$2\spbound$ divides $m$, then the corresponding $(n,\spbound)$-sparse polynomial 
evaluation code has minimum distance  $\delta
= \lfloor \frac{n}{2\spbound}\rfloor$. 
\end{theorem}

The following  proof is adapted from ~\cite[\S 2.1]{CKP12}.
\begin{proof}
Let $\bar{0}$ denote the zero vector of length $\spbound-1$.
Consider two infinite sequences :
\begin{eqnarray*}
x&=&(\bar{0},1,\bar{0},1,\ldots)\\
y&=&(\bar{0},1,\bar{0},-1,\bar{0},1,\bar{0},-1,\ldots)
\end{eqnarray*}
formed by the repetition of their first $2\spbound$ values and the corresponding
vectors $x^{(n)},y^{(n)} \in \KK^n$ and $x^{(m)},y^{(m)} \in \KK^m$  formed by
respectively the first $n$ and first $m$ values of these sequences. 
The sequence $x$ is generated by $z^\spbound-1$ and $y$
by $z^\spbound+1$, both are $m$ periodic as $2\spbound$ divides $m$.
Lastly, let $\hat x^{(m)}
= \text{DFT}_\alpha^{-1}(x^{(m)})=\frac{1}{m}\text{DFT}_{\alpha^{-1}}(x^{(m)})$. 
From Blahut's theorem, $\hat x^{(m)}$ has Hamming weight $\spbound$. 
By identification between $\KK^m$ and $\KK[z]_{<m}$, $\hat x^{(m)}$ corresponds to
a polynomial $f_x$ of degree less than $m$ and sparsity $\spbound$. 
Hence $x^{(n)} =
(f_x(\alpha^0),f_x(\alpha^1),\dots,f_x(\alpha^{n-1}))$ is a code word of an
$(n,\spbound)$-sparse evaluation code. Similarly $y^{(n)}$ is also a code-word.
More precisely one verifies that 
\begin{eqnarray*}
f_x(z)&=&\frac{1}{T}\sum_{i=0}^{T-1}z^{2i\frac{m}{2T}} %ELK3 add 2
          = \frac{1}{T}\frac{z^m-1}{z^{\frac{m}{T}}-1}, \\
f_y(z)&=&\frac{-1}{T}\sum_{i=0}^{T-1}z^{(2i+1)\frac{m}{2T}} %ELK3 add 2: +\frac{m}{2T}
          = \frac{-z^{\frac{m}{2T}}}{T}\frac{z^m-1}{z^{\frac{m}{T}}-1}. 
\end{eqnarray*}

%% For $m=2\spbound$, $x_1$ and $x_2$ are indeed code-words of an
%% $(n,\spbound)$-sparse evaluation code:
%% for $x_1$ we have $f_1 = (1/T) \sum_{i=0}^{T-1} z^{2i} = (1/T)(z^{2\spbound} - 1)/(z^2 - 1)$
%% and for $x_2$ we have $f_2 = -(1/T) \sum_{i=0}^{T-1} z^{2i + 1} = -(z/T) (z^{2\spbound} - 1)/(z^2 - 1)$.
%% For $\ell \not\equiv 0 \pmod{T} \Leftrightarrow \alpha^\ell \ne \pm 1$,
%% we have $f_1(\alpha^\ell) = f_2(\alpha^\ell) = 0$
%% by $(\alpha^\ell)^{2\spbound} = (\alpha^m)^\ell = 1$.
%% For $\ell = 2T, 4T, 6T, \ldots \Leftrightarrow \alpha^\ell = 1$ we have $f_1(1) = 1$ and $f_2(1) = -1$
%% while for
%% $\ell = T, 3T, 5T, \ldots  \Leftrightarrow \alpha^\ell = -1$ we have $f_1(-1) = f_2(-1) = 1$.
Since $x^{(n)}$ and $y^{(n)}$ differ by exactly $\lfloor \frac{n}{2\spbound} \rfloor$
values, this is an upper bound on the minimum distance~$\delta$.

Now consider any pair of distinct code-words $x$ and $y$ and consider their
$\lfloor\frac{n}{2\spbound}\rfloor$ sub-vectors
$$
\begin{array}{ll}
x^{(1)}=(x_1,\dots,x_{2\spbound}), & y^{(1)}=(y_1,\dots,y_{2\spbound})\\
x^{(2)}=(x_{2\spbound+1},\dots,x_{4\spbound}),& y^{(2)}=(y_{2\spbound+1},\dots,y_{4\spbound})\\
\vdots & \vdots\\
x^{(\lfloor\frac{n}{2\spbound}\rfloor)},&  y^{(\lfloor\frac{n}{2\spbound}\rfloor)}\\
\end{array}
$$

If for some $i$,  $x^{(i)}=y^{(i)}$ then the vector $z^{(i)}=x^{(i)}-y^{(i)}$ is all
zero and is the evaluation of a less than $2\spbound$-sparse polynomial $f-g$. Solving the $2\spbound\times
2\spbound$ corresponding Vandermonde system yields $f=g$ which is a contradiction.
Hence $x$ and $y$ differ in at least $\lfloor\frac{n}{2\spbound}\rfloor$ positions, and
consequently $\delta=\lfloor\frac{n}{2\spbound}\rfloor$.
\end{proof}

\myparagraph{Unique decoding}
There exists an algorithm that does unique decoding of such codes up to half the
minimum distance: the Majority Rule \BMA~\cite{CKP12}.
It simply consists in running a \BMA on each of the
$\lfloor\frac{n}{2\spbound}\rfloor$ contiguous sub-sequences
$x^{(i)} = (x_{2\spbound i},\ldots,x_{2\spbound(i+1)-1})$
of the received word $x$. If
$E< \lfloor\frac{n}{2\spbound}\rfloor/2$ errors occurred, then the generator occurring
with majority will be the correct one. We refer to~\cite{CKP12} for further
explanations on how to then recover the correct code-word using sequence clean-ups.
Equivalently, this algorithm guaranties to find the unique code-word provided that  $E$
errors occured whenever $n\geq 2\spbound(2E+1)$.
This decoding requires $\lfloor n/2\spbound \rfloor$ executions of \BMA.

\myparagraph{List decoding}
Following the same idea, one remarks that if  $n\geq 2\spbound(E+1)$ then necessarily,
one sub-sequence $x^{(i)}$ has to be clean of errors and the list of all
$\lfloor\frac{n}{2\spbound}\rfloor$ generators contains the correct one.
This makes a trivial list decoding algorithm up to the minimum distance (see \cite{CKP12} for further
details on how to recover the code-word using sequence clean-ups).

In order to further reduce the bound $n\geq 2\spbound(E+1)$ (or equivalently increase
the decoding radius above $\frac{n}{2T}$), we will study in
Section~\ref{sec:affineseq} an alternative list decoding algorithm.
Beforehand, we want to address a common remark on the choice of the sub-vectors
used for the unique and list decoding above. 
\begin{remark}\label{rem:allsubseq}
Instead of partitioning the received word into $n/(2T)$ %ELK3 add () here and elsewhere
disjoint sub-vectors, one would hope to find more
error-free sequences by considering all 
%CP Revision Ref2
%$n-2T$ 
$n-2T+1$ 
sub-vectors of the form
$(x_{i},\ldots, x_{i+2T-1})$. This will very likely allow to decode more errors
in many cases (as will be illustrated in Figure~\ref{fig:randomerr:allsubseq}), but the
worst case configuration (see proof of Theorem~\ref{th:mindist}) remains
unchanged. 
Note that the majority rule based unique decoding still works under
the same conditions: at most $2TE$ sub-sequences will contain an error, hence 
%CP Revision Ref 2
%% at least
%% $n-2T-2TE$ sequences are correct and form a majority as soon as
%% $n-2T-2TE\geq (n-2T)/2$ which is $n\geq 2T(2E+1)$.
a majority of subsequences will be correct as soon as $4TE<n-2T+1$, which is
$n\geq 2T(2E+1)$.
In terms of complexity, the number of arithmetic operations required for both unique and
list decoding algorithms in~\cite{CKP12} is $\GO{n^2}$ ($n/(2T)$ runs of \BMA on sequences of
length $2T$, and $\GO{n/(2T)}$ calls to the sequence clean-up, each of which costs $\GO{nT}$).
Now the above variant requires to inspect $n-2T$ sub-sequences instead of $n/(2T)$ and
the complexity becomes 
%CP Revision ref1
%$\GO{nT^2+n^2T}$.
$\GO{n^2T}$ (as $T=o(n)$).
\end{remark}

\section{Affine sub-sequences}
\label{sec:affineseq}
Consider a sequence $(a_0,\ldots,a_{n-1})$ of evaluations of a $t$-sparse
polynomial $f(z)=\sum_{j=1}^t c_jz^{e_j}$, with $E$ errors. 
In our previous work, we used to search for sub-sequences of the form
$(a_i,\ldots,a_{i+k-1})$ formed by $k$ consecutive elements that did not contain
any error. If such a sequence could be found with $k=2t$, then applying
Blahut/Ben-Or/Tiwari %ELK3 was: Ben-Or/\-Tiwari's
%CP2: reverted to Ben-Or/\-Tiwari's
algorithm on it recovers the polynomial $f$ and makes the
decoding possible.
We now propose to consider all length $k$ sub-sequences
in arithmetic progression:
\[
(a_{r},a_{r+s},a_{r+2s},\ldots,a_{r+(k-1)s}) \text{ where } r+(k-1)s< n,
\]
that  will be called affine index sub-sequences or more conveniently affine
sub-sequences. In the remaining of the text, $k$ will 
denote the length of the sub-sequence. We will consider the general case
where $k$ can be any positive integer, not necessarily even.
\begin{lemma}
If $\gcd(s,m)=1$ and 
%CP Revision Ref1
%$k\geq 2E$,
$k\geq 2t$,
 then such a sub-sequence with no error is sufficient to
recover $f$.
\end{lemma}
\begin{proof}
  Let $\beta = \alpha^s$ and $g(z)=f(z\alpha^{r})$.
  Note that $\deg g = \deg f$ and $g$ is also $t$-sparse with the same monomial
  support as $f$.
  If $\gcd(s,m)=1$ then $\text{order}(\beta)\geq m$.
  Then the sub-sequence 
$(a_{r},a_{r+s},a_{r+2s},\ldots,a_{r+(k-1)s}) = \allowbreak
(f(\alpha^r),f(\alpha^{r+s}),f(\alpha^{r+2s}),\ldots,f(\alpha^{r+(k-1)s})) = \allowbreak
(g(\beta^0), \allowbreak g(\beta^1),g(\beta^2), \ldots, \allowbreak g(\beta^{k-1})
)  
$
is formed by evaluations of $g$ in $k$ consecutive powers of an element $\beta$ of order
greater than $ m \geq \deg g$.
One can thus compute $g=\sum_{j=1}^td_jz^{e_j}$ using Blahut/Ben-Or/Tiwari 
%ELK3 was: Ben-Or/\-Tiwari
%CP2: reverted to Ben-Or \& Tiwari
algorithm on this sub-se\-quence.
The coefficients of $f$ are directly deduced from that of $g$: $c_j = d_j\alpha^{-re_j}$.
\end{proof}

\begin{example}
Let $t=2$, and consider a sequence of $n=9$ evaluations $(a_0,a_1,\ldots,a_8)$.
Then $E=1$ is the maximal number of errors that the list decoding
of~\cite{CKP12} can decode as it requires that $n\geq 2t(E+1)$.
Indeed if two  errors occurred e.g. on elements $a_3$  and $a_7$, 
there is no contiguous sub-sequence of length $2t=4$ free of error, thus making
the latter decoding fail.
Now consider the sub-sequence $(a_0,a_2,a_4,a_6)$. It is free of error and is
formed by evaluations of $f(z)$ in the four consecutive powers of
$\beta=\alpha^2$.  Blahut/Ben-Or/Tiwari
%ELK3 was: Ben-Or/Tiwari
%CP2: reverted to Ben-Or \& Tiwari
algorithm applied on this sequence will reveal $f$.
\end{example}

\myparagraph{A list decoding algorithm}

This results in a new list decoding algorithm:

\par\noindent\hangindent=1.3em\hangafter=1 %
1. For each affine sub-sequence $(a_{r},a_{r+s},\dots,a_{r+s(k-1)})$ compute a
generator $\Lambda_{r,s}$ with the \BMA
\par\noindent\hangindent=1.3em\hangafter=1 %
2. (Optional heuristic reducing the list size) For each $\Lambda_{r,s}$, run the sequence clean-up
of~\cite{CKP12} and discard it if it can not generate the sequence with less
than $E$ errors, for some  bound~$E$ on the number of errors. 
\par\noindent\hangindent=1.3em\hangafter=1 %
3. For each remaining generator $\Lambda_{r,s}$, apply Blahut/Ben-Or/\-Tiwari
algorithm to recover the associated sparse polynomial $f_{r,s}$.
\par\noindent\hangindent=1.3em\hangafter=1 %
4. Return the list of the $f_{r,s}$.
\newline

%% There are $n/k$ possible values for $s$ and for each value of $s$, there are
%% $s\frac{n}{ks}=n/k$ choices for $r$ when restricting the
%% search to disjoint sub-sequences.
A first approach is to explore all sub-sequences for any value of
$s\in\{1\dots \lfloor n/k\rfloor\}$ and  $r\in\{0\dots n-(k-1)s-1\}$.
This amounts to $\GO{n^2/k}$ sub-sequences.
A second approach, applying Remark~\ref{rem:allsubseq} considers all values
for $s\in\{1\dots n/k\}$ but then for each $s$ only considers the disjoint
sub-sequences with $s\frac{n}{ks}=n/k$ choices for $r$. This amounts to
$\GO{n^2/k^2}$ sub-sequences.
For each sub-sequence, corresponding to a pair $(s,r)$, Blahut/Ben-Or/Tiwari
algorithm is run in $\GO{k^2}$ (\BMA and solving the transpose
Vandermonde system~\cite{Zi90}). The
optional sequence clean-up heuristic adds an $\GO{nk}$ term.
Overall, the complexity of the second approach amounts to the same $\GO{n^2}$
estimate, as the list decoding of~\cite{CKP12}. The
additional overhead of \GO{n^3/k} when the sequence clean-up heuristic is used
also remains identical.
%Without the clean-up, the returned list has size bounded by $n^2/k^2$.
In the first approach, ignoring Remark~\ref{rem:allsubseq}, these complexity
estimates are multiplied by a factor $k$.

%% Applying Remark~\ref{rem:allsubseq}, $r$ can now take $n/s$ different
%% values. The time complexity increases by a factor $k$ : $n^3+n^2k$.

We implemented the affine sub-sequence search and computed its rate of success
in finding a clean sequence for various values of $E$ and $k$. The error
locations are uniformly distributed. We report in Figures~\ref{fig:randomerr}
and ~\ref{fig:randomerr:allsubseq} the average rate of success over $10\,000$
samples for each value of the pair $(E,k)$.
Figure~\ref{fig:randomerr} uses the search restricted to disjoint
sequences.
\begin{figure}[h!]
\begin{center}
\includegraphics[width=\graphl]{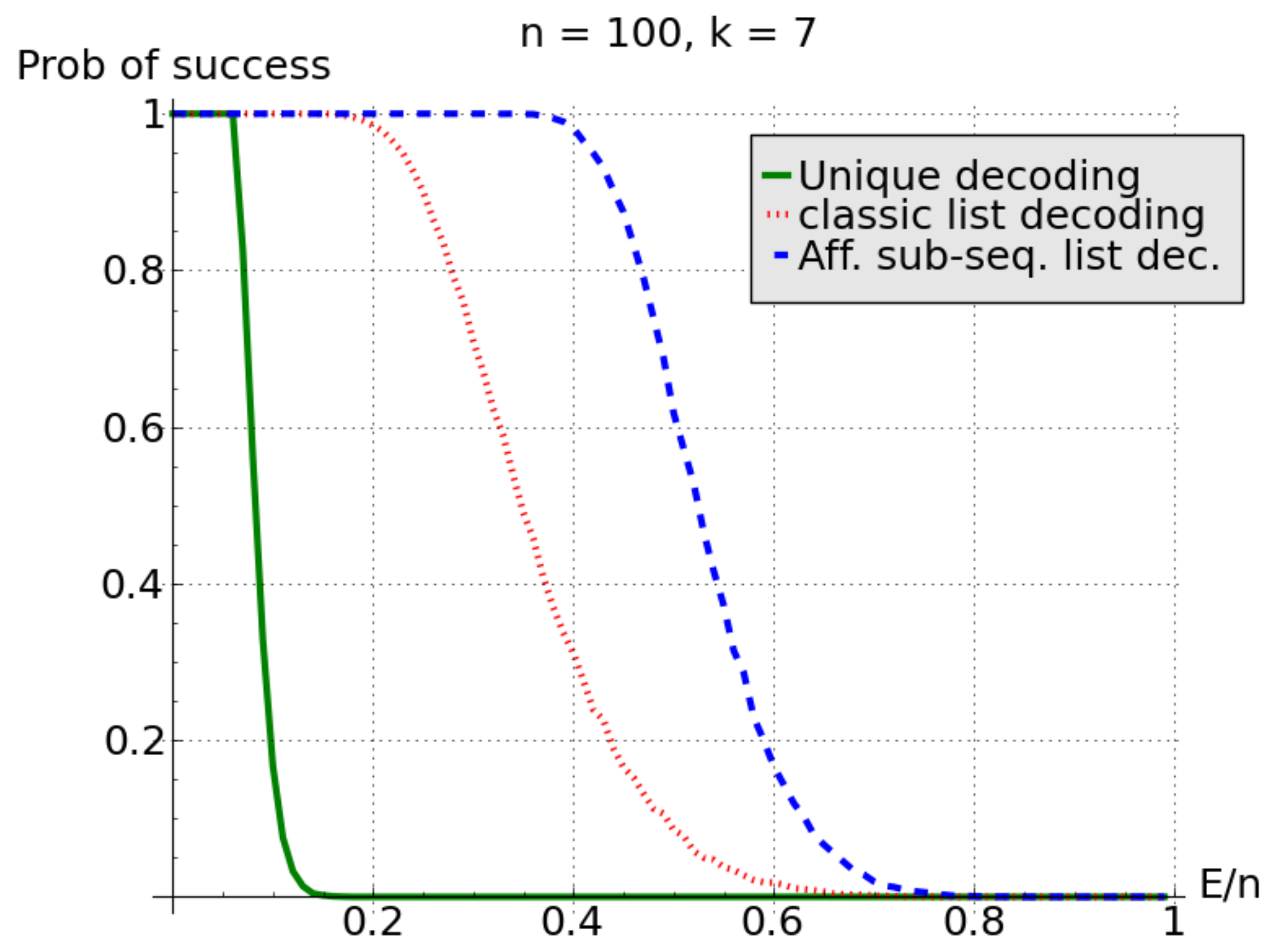}
\end{center}
\caption{Success rate for unique, standard list decoding and affine
sub-sequence list decoding. Only disjoint sub-sequences are considered.}
\label{fig:randomerr}
%Each point has been computed by $10\,000$ uniformly
%random error vectors of weight $0\leq E<n$.}
\end{figure}
whereas Figure~\ref{fig:randomerr:allsubseq} shows the improvement brought by
considering all sub-sequences as proposed in Remark~\ref{rem:allsubseq}. 
\begin{figure}[h!]
\begin{center}
\includegraphics[width=\graphl]{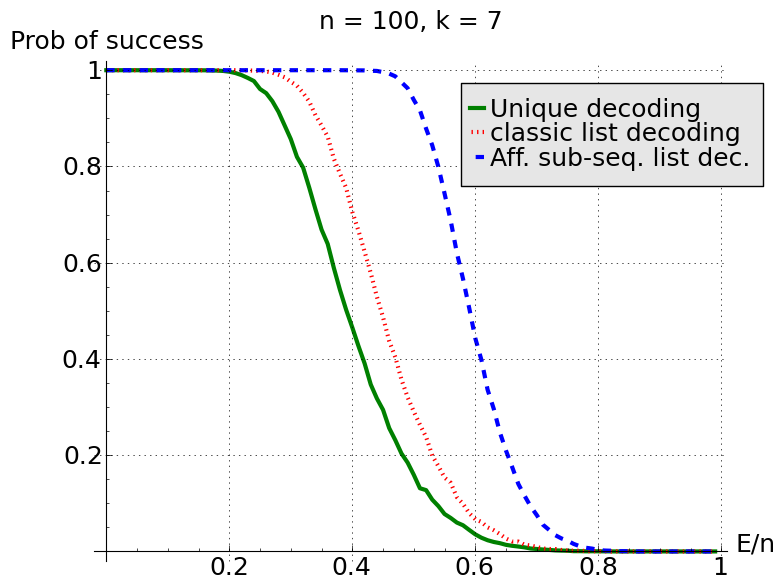}
\end{center}
\caption{Success rate for unique, standard list decoding and affine
sub-sequence list decoding. All sub-sequences are considered.}
\label{fig:randomerr:allsubseq}
%Each point has been computed by $10\,000$ uniformly
%random error vectors of weight $0\leq E<n$.}
\end{figure}
Again this improvement is important in practice, at the expense of a higher
computational complexity for the decoder, but does not improve the
%CP Revision Ref 3
unique 
decoding radius in the worst case.
\vspace{1em}

%%%%%%%%%%%%%%%%%%%%%%%%%%%%%%%%%%%%%%%%%%%%%%%%%%%%%%%%%%%%%%%%%%%%%%%%%%%
\section{Worst case decoding radius}

We now focus the worst case analysis: finding estimates on the maximal decoding
radius of the  affine sub-sequence algorithm.
More precisely, we want to determine for fixed $E$ and $t$, the smallest
possible length $n$, such that for any error vector of weight up to $E$, there
always exist at least one affine sub-sequence of length $2t$ with no error. 
%In order to quantify the improvement of this technique, 
This is stated in 
%We set it as
Problem~\ref{pb:affineseq} in the more general setting where the length  $k$ of
the error free sequence need not be even.%is not restricted to even integers.
\begin{problem}\label{pb:affineseq}
Given $k,E \in \Z_{>0}$, find the smallest $n\in \Z_{>0}$ such that for
all subsets %ELK3 was: any subset
$S\subset \{0,\ldots,n-1\}$
with $E$ elements, that is, $|S| = E$, %ELK3 of $E$ distinct integers, 
\begin{multline}%ELK3
\label{eq:affineseq}
%ELK3 \begin{split}
  \exists r\in \Z_{\geq 0}, \exists %ELK3 add
   s\in \Z_{>0}
\text{ with } r+s(k-1) \le n-1 \colon %ELK3 was: \text{ s.t. }\\
%ELK3    \{r+is:0\leq i \leq k-1\} \cap S = \emptyset 
%ELK3 better definition?
\\
\forall i\text{ with } 0 \le i \le k-1 \colon r+is \not\in S.
%ELK3 \end{split}
\end{multline}
\end{problem}
We will denote by $n_{k,E}$ the minimum %ELK3 add
solution to Problem~\ref{eq:affineseq}. %ELK3 was: this problem.

In some cases, the affine sub-sequence technique does not help improving
the former bound $n\geq k(E+1)$, not even by saving a single evaluation point.
\begin{example}\label{ex:K5E3}
For $k=5$ and $E=3$, the worst case configuration (errors on $a_4, a_9$ and $a_{14}$)
requires $n_{5,3}=20=k(E+1)$ values to find $k$ %ELK3 was $t$ ???
consecutive clean values.
\vspace{1em}
%\begin{table*}[H!]
%\caption{
%}
%\renewcommand{\arraystretch}{0.9}

%\centering

%\setlength{\arraycolsep}{.1\arraycolsep}
{%\center
\noindent
 \begin{tabular}{@{\extracolsep{-3mm}}llllllllllllllllll}
\toprule
    $a_0$ &
    $a_1$ &
    $a_2$ &
    $a_3$ &
    $\mathbf{a_4}$ &
    $a_5$ &
    $a_6$ &
    $a_7$ &
    $a_8$ &
    $\mathbf{a_9}$ &    
 %   $a_{10}$ 
%\\
%\bottomrule
%\end{tabular} $\ldots$\\
%\vspace{1em}
%\hfill 
&
$\ldots$
&
%\begin{tabular}{@{\extracolsep{-3mm}}lllllllll}
%\toprule
%    $a_{11}$ &
%    $a_{12}$ &
%    $a_{13}$ &
    $\mathbf{a_{14}}$ &
    $a_{15}$ &
    $a_{16}$ &
    $a_{17}$ &
    $a_{18}$ &
    $a_{19}$ \\
\bottomrule
  \end{tabular}
%\end{table*}
\vspace{5pt}

\noindent
  \begin{tabular}{@{\extracolsep{-3mm}}llllllllll}
\toprule
    $a_0$ & $a_2$&$\mathbf{a_4}$&$a_6$&$a_8$&$a_{10}$ & $a_{12}$ & $\mathbf{a_{14}}$& $a_{16}$ & $a_{18}$\\
    $a_1$ & $a_3$& $a_5$& $a_7$&$\mathbf{a_9}$& $a_{11}$ &$a_{13}$& $a_{15}$ &$a_{17}$& $a_{19}$\\
\bottomrule
  \end{tabular}
\vspace{5pt}

\noindent
\begin{tabular}{@{\extracolsep{-3mm}}lllllll}
\toprule
    $a_0$ & $a_3$ &$a_6$ &$\mathbf{a_9}$ & $a_{12}$ & $a_{15}$ & $a_{18}$\\
    $a_1$ & $\mathbf{a_4}$ & $a_7$ &  $a_{10}$& $a_{13}$ & $a_{16}$ & $a_{19}$ \\
    $a_2$ &  $a_5$ &$a_8$ & $a_{11}$ & $\mathbf{a_{14}}$ & $a_{17}$\\
\bottomrule
  \end{tabular}
%\vspace{3pt}
%\hspace{2em}
%
\hfill
\begin{tabular}{@{\extracolsep{-3mm}}lllll}
\toprule
$a_0$ & $\mathbf{a_4}$ & $a_8$ & $a_{12}$ & $a_{16}$ \\
$a_1$ & $a_5$ & $\mathbf{a_9}$ & $a_{13}$ & $a_{17}$ \\
$a_2$& $a_6$& $a_{10}$& $\mathbf{a_{14}}$& $a_{18}$\\
$a_3$&$a_7$&$a_{11}$ &$a_{15}$&$a_{19}$\\
\bottomrule
  \end{tabular}
\vspace{1em}

}

But for $E=4$, one verifies that $n=21$ suffices to ensure that a length $5$
subsequence will always be found. In particular, in the previous configuration,
placing the fourth error on $e_{19}$ leaves the subsequence
$(a_0,a_5, a_{10},\allowbreak a_{15}, a_{20})$ untouched.
%\vspace{1em}
\begin{center}
%\noindent
 \begin{tabular}{@{\extracolsep{-1mm}}lllll}
\toprule
$a_0$ &$a_5$ & $a_{10}$&$a_{15}$& $a_{20}$\\
$a_1$ & $a_6$& $a_{11}$&$a_{16}$ \\
$a_2$& $a_7$ & $a_{12}$&$a_{17}$ \\
$a_3$& $a_8$ & $a_{13}$& $a_{18}$\\
$\mathbf{a_4}$ &$\mathbf{a_9}$& $\mathbf{a_{14}}$&$\mathbf{a_{19}}$\\
\bottomrule
  \end{tabular}
%\end{table*}
\end{center}
\end{example}

We report in Table~\ref{tab:bestn} and Figure~\ref{fig:bestn} the values of $n_{k,E}$ 
for all typically small values of $E$ and $k$ computed by exhaustive search. We ran a Sage
program\footnote{The code is
available \url{http://membres-liglab.imag.fr/pernet/Depot/ldsic.sage}.} for
about 7    days on  24 cores of an Intel E5-4620 SMP machine.
\begin{table*}
%\small
  \begin{center}
    \begin{tabular}{l|rrrrrrrrrrrrrrrr}
      \toprule
E     & 0  & 1  & 2  & 3  & 4  & 5  & 6  & 7  & 8  & 9   &  10& 11 & 12 &  13 & 14 & 15 \\
\midrule
%$k=2$ &  2 &  3 &  4 &  5 &  6 &  7 &  8 &  9 & 10 & 11  & 12 & 13 & 14 & 15  & 16 & 17\\
$k=3$ &  3 &  6 &  7 &  8 & 10 & 12 & 15 & 16 & 17 & 18  & 19 & 21 & 22 & 23  & 25 & 27\\
$k=4$ &  4 &  7 & 11 & 12 & 14 & 16 & 18 & 20 & 22 & 24  & 26 & 29 & 31 & 32  & 35 & 36\\
$k=5$ &  5 & 10 & 15 & 20 & 21 & 22 & 23 & 26 & 30 & 32  & 35 & 40 & 45 & 46  &
47 & 48\\
$k=6$ &  6 & 11 & 16 & 21 & 27 & 28 & 30 & 31 & 34 & 38  & 42 & 43 & 47 & 52\\
$k=7$ &  7 & 14 & 21 & 28 & 35 & 42 & 43 & 44 & 45 & 47  & 49 & 54 & 58\\
$k=8$ &  8 & 15 & 22 & 29 & 36 & 43 & 51 & 52 & 53 & 55  & 57 & 60\\
$k=9$ &  9 & 18 & 25 & 32 & 39 & 46 & 53 & 58 & 59 & 62  & 66 & 72\\
$k=10$& 10 & 19 & 29 & 34 & 41 & 48 & 55 & 62 & 65 & 69  & 74\\
$k=11$& 11 & 22 & 33 & 44 & 55 & 66 & 77 & 88 & 99 & 110 & 111 & 112\\
$k=12$& 12 & 23 & 34 & 45 & 56 & 67 & 78 & 89 &100 & 111 & 123 & 124\\
$k=13$& 13 & 26 & 39 & 52 & 65 & 78 & 91 & 104&117 & 130 & 143 & 156 \\
\bottomrule
    \end{tabular}
    \caption{Solution $n_{k,E}$ to Problem~\ref{pb:affineseq} for the first values of $k$ and $E$}\label{tab:bestn}
  \end{center}
\end{table*}

\begin{figure}[h!]       
%\begin{center}
  \includegraphics[width=\graphl]{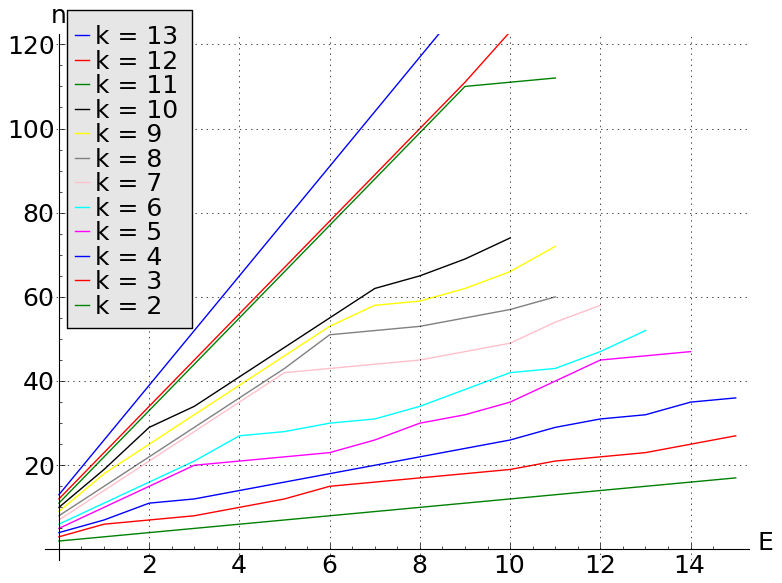}
  \caption{Solutions to Problem~\ref{pb:affineseq} for the first values of $k$
  and $E$}\label{fig:bestn}
%\end{center}
\end{figure}

\begin{figure}[h!]
%\begin{center}
        \includegraphics[width=\graphl]{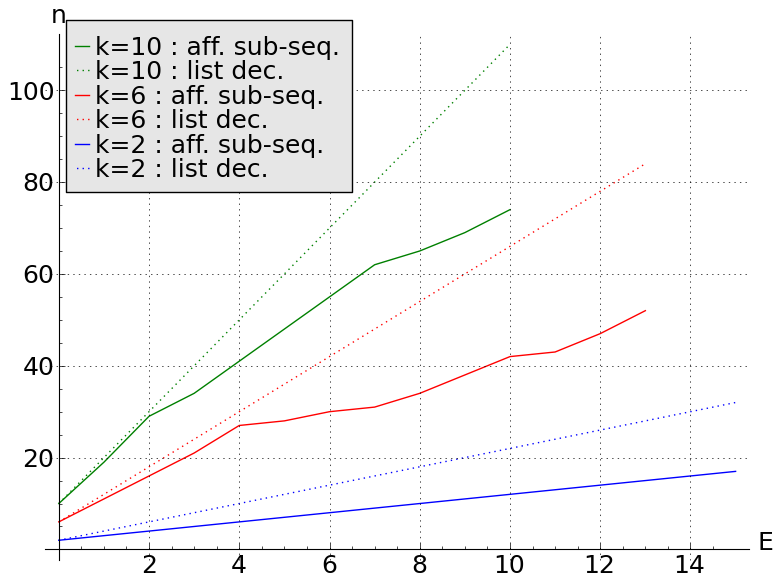}
%\end{center}
\caption{Improvement of the affine subsequence approach over the list
decoding of~\cite{CKP12}}
\label{fig:comp}
\end{figure}

In particular Figure~\ref{fig:comp} shows the improvement of the affine
sub-sequence technique over the previous list decoding algorithm, requiring
$n= k(E+1)$, for some even values of the sub-sequence length $k$.
%
%%   \begin{figure}[htbp]
%% \begin{center}
%%   \includegraphics[width=.8\textwidth]{slopeFitting}
%%   \caption{Fitting the slope of the curves beyond $E=t-2$ to 2: plot of
%%     $n-2E$ for various values of $E$ and $t$} 
%% \end{center}
%%   \end{figure}
%
These data indicate that the optimal value for the length $n$ is improved in
almost any case 
except when $k$ is prime and  $E<k-1$, as in
Example~\ref{ex:K5E3}. Lemma~\ref{lem:prime} states more precisely at which 
condition the new algorithm does not improve the value $n=k(E+1)$ of the former
list decoding.

\begin{lemma}\label{lem:prime}
%ELK3 rewrite: Let $g>1$ be the smallest divisor of $k$, then $n_{k,E}=k(E+1)$ if and only if $  E+1 <g$.
We have for the minimum solution $n_{k,E}$ of Problem\/~{\upshape \ref{pb:affineseq}:}
$n_{k,E} \le k(E+1)$, with equality $n_{k,E}\text{=}k(E\text{$+$}1)$
if and only if $E+2\le g$, where 
$g$ denotes the smallest 
%non-trivial $(g \ge 2)$ %CP2: non trivial => g != k  which is incorrect
prime factor of $k$.
%, satisfies $g \ge E+2$.
%$k$ is co-prime with $E+1$ 
%and $E<k-1$.
\end{lemma}
In particular, this implies that for
even %ELK3 add
$k=2t$, the new list decoder performs always better.
\begin{proof}
%  Let $k$ be prime, $E<k-1$ and $S=\{k-1, 2k-1, \ldots Ek-1\}$. 
    %
%Consider the set of $k$ integers $\{Ek, Ek+1, \ldots, Ek+k-1\}$. It contains no
%element of $S$, hence 
Let $n=k(E+1)$. Splitting $\{0,\ldots,n-1\}$ into $E+1$ contiguous disjoint sets
$V_i$ of $k$ elements 
shows that no subset of $E$ elements of  $\{0,\ldots,n-1\}$ can intersect all
of the $V_i$'s at the same time.
Hence $n_{k,E} \leq (E+1)k$. 

%We will show that no smaller value for $n$ is possible.
We will denote by $P_{r,s}$ the arithmetic progression 
%ELK3 was: subset 
%CP2: linear ->arithmetic
$\{r$, $r+s$, $\ldots$, $r+(k-1)s\}$.

Suppose $n=k(E+1)
= n_{k,E}$. %ELK3 add
Then there is a subset $S$ of $E$ elements of
$\{0,\dots,n-2\}$ that intersects all $P_{r,s} \subset \{0,\dots,n-2\}$. We will show that
$S=\{k-1,2k-1,\dots,Ek-1\}$. 
Indeed, as the $E$ segments %ELK3 was: subsets
$P_{ik,1}$ for $0\leq i \leq E-1$ are disjoint,
 each of them must contain exactly one element of $S$. Hence,
 $\{Ek,\dots,Ek +k-2\}\cap S =\emptyset$, and therefore $Ek-1 \in S$, otherwise
 $P_{Ek-1,1}\cap S=\emptyset$. By the same argument, we deduce iteratively that
 $ik-1 \in S$ for all $i\leq E-1$. It follows that $E+1<g$, otherwise $n\geq
 kg$ and $P_{0,g}\subset \{0,\dots,n-2\}$ but $P_{0,g} \cap
 S=\emptyset$ since $P_{0,g} %ELK3 was: \pmod replaced elsewhere
\bmod g = \{0\}$ and $S \bmod g = \{k-1\}$.
%%  Lastly, suppose $\gcd(E+1,k)=g>1$ and write $E+1=gp, k=gq$.
%% As $S \mod p = \{p-1\}$, the sequence $P_{0,p}  \subset \{0,n-2\}$ does not
%% intersect $S$ which is absurd.

%====
Now suppose $E+1<g$.
We will show that $S=\{k-1,2k-1,\dots,Ek-1\}$ intersects all
$P_{r,s}\subset \{0,
1,\ldots, %ELK3 add
k(E+1)-2\}$ from which we shall deduce that
$n_{k,E}=k(E+1)$.

First note that $s<g$: otherwise
$r+g(k-1)\leq r+s(k-1)
%\leq n_{k,E}-1 %ELK3 add % CP2: removed as n_{k,E} is not yet defined
\leq k(E+1)-2 %ELK3 was -2 % CP2: reverted: indeed this is -2 (see 2 lines above)
\leq k(g-1)-2$ would imply
$r+k \le g-2$, %ELK3 was $r\leq g-k-2<0$ %CP2 reverted to g-k-2 though not
               %necessary
which is absurd.
%
         %%  and suppose that
%% $\exists r\in \Z_{\geq 0}, s \in \Z_{>0},$ such that 
%% $P_{r,s} \cap S =\emptyset$.
%% We have $r+(k-1)s\leq n-1<(E+1)k-1$. If $s\geq E+2$, then $r+(k-1)(E+2)<(E+1)k-1$ which
%% implies $r<E-k+1<0$ which is a contradiction. Consequently $s\leq E+1<k$. 
Hence $s$ is co-prime with $k$. %ELK3 remove: and its order modulo $k$ is $k-1$.
Therefore $P_{r,s} \bmod k  = \{0,1,\ldots,k-1\}$,
%and therefore $P_{r,s}\cap S \neq \emptyset$.
% which concludes the proof.
%CP2: argument below was missing
and 
$ \exists j\leq k-1$ and $q\in \Z_{\geq 0}$ such that 
%\left\{
% \begin{array}{ll}
$r+js = k-1 + qk = (q+1)k-1$.
% \end{array}
% \right.
As $r+js\leq(E+1)k-2$ we have $q<E$.
Hence $r+js \in  S$ and finally $P_{r,s}\cap S \neq \emptyset$.
 %% Therefore $r+js=(q+1)k-1 \geq (E+1)k-1$, but we also have  $r+js \leq n-1 \leq
 %% (E+1)k-2$
 %% %. Thus $k(E+1)\leq k(E+1)-1$
 %%  which is absurd.
%% Consequently $n\geq k(E+1)$.
\end{proof}
 
Note that the  solution $n_{k,E}$ is also strictly increasing in both $k$ and $E$:
%\begin{lemma}
$\forall k>0, E\geq0 \
\left\{\begin{array}{l}n_{k+1,E} > n_{k,E}\\ n_{k,E+1} > n_{k,E}\end{array}\right.$.
%\end{enumerate}
%\end{lemma}
This implies that $n_{k,E}$ always verifies 
$p(E+1) \leq n_{k,E}\leq k(E+1)$ where $p$ is the prime previous to $k$ for
$E<p-1$.
%% \begin{cor}
%% Let $p,q$ be two prime numbers such that $p<k<q$. Then for any $E<p-1$,
%% \[
%% p(E+1) < n_{k,E} < q(E+1).
%% \]
%% \end{cor}
%% Lastly  Bertrand's postulate states that there is always a
%% prime $p$ between $n>3$ and $2n-2$. 
Corollary~\ref{cor:lowerbound} gives a lower bound on $n_{k,E}$.
\begin{cor}\label{cor:lowerbound}
%\begin{enumerate}
%\item
If $E<\frac{k+1}{2}$, then $\frac{k+1}{2}(E+1)+1<n_{k,E}$
for any $k>0$.

%\item If $E<k$, then $n_{k,E}<2k(E+1)$ 
%\end{enumerate}
\end{cor}
\begin{proof}
By Bertrand's postulate~\cite{Cheb52,Ram19}, if $k\geq 6$, there exist a prime $p$ such that 
\[
\left\lceil\frac{k+1}{2}\right\rceil <p< 2\left\lceil\frac{k+1}{2}\right\rceil-2
\leq k.
\]
Therefore $\frac{k+1}{2}(E+1)<p(E+1)=n_{p,E}< n_{k,E}$ if $E<\frac{k+1}{2}$. One
verifies the cases $k \leq 5$ on table~\ref{tab:bestn}.
%Similarly, there exist a prime $q$ such that $k<q<2k$, hence
%$n_{k,E}<n_{q,E}=q(E+1)<2k(E+1)$ if $E<k$.
\end{proof}

However, for larger values of $E$, Figure~\ref{fig:bestn} suggests that $n_{E,k}$
increases at a much lower rate.
We will now focus on the asymptotic behavior of $n_{k,E}$.
In order to find a lower bound on the value of $n$ that solves
Problem~\ref{pb:affineseq}, we
%CP Revision Ref 3
% consider the dual Problem:
%~\ref{pb:dual}. 
%% \begin{problem}\label{pb:dual}
%% Given $k,E\in \Z_{>0}$, find the maximal $n\in \Z_{>0}$ such that for any subset
%% $S\subset \{0,\ldots, n-1\}$ of $E$ distinct integers (\ref{eq:affineseq}) does
%% not hold.
%% \end{problem}
%In the following, the solution to this problem will be denoted by $N_{k,E}$.
%
%Note that $n_{k,E_i}=N_{k,E_i}+1$.
%
%We 
will construct by induction a subset $S$ producing a large value for $n$.
%that results in a lower bound on the solution of 
%CP Revision Ref3
%Problem~\ref{pb:dual} and hence to that of 
%Problem~\ref{pb:affineseq}. 
It is a generalization of the worst case
error vector of Lemma~\ref{lem:prime}.

For all $i\in \Z_{\geq 1}$ let $m_i=k^{i-1}(k-1)$ and define the error vector $v_i$ by the recurrence 
\[\left\{
\begin{array}{l}
  v_1=(\underbrace{0,\ldots,0}_{k-1 \text{ times}}) \in \KK^{m_1}\\
  v_{i+1} = (\underbrace{v_{i},\overbrace{1,\ldots,1}^{k^{i-1}
      \text{ times}},\ldots,v_{i},\overbrace{1,\ldots,1}^{k^{i-1} \text{ times}}}_{k-1
    \text{ times}} ) \in \KK^{m_{i+1}}
\end{array}\right.
\]
Lastly, define $w_i$ as the vector $v_i$ without its trailing $k^{i-2}+k^{i-3}+\dots+1$ ones.
$w_i$ has length 
$$n_i=k^i-k^{i-1}-\dots-1= k^i - \frac{k^i-1}{k-1}=\frac{(k-2)k^i+1}{k-1}.$$ 
The Hamming weight of $v_i$ satifies $w_H(v_1)=0$ and
$w_H(v_{i+1})=(k-1)(w_H(v_{i})+k^{i-1})$ which solves into $w_H(v_i) =
k^{i}-k^{i-1}-(k-1)^{i} = m_i-(k-1)^{i}$.
Finally, $w_i$ has weight $E_i=n_i-(k-1)^{i}$.
%% For all $i\in \Z_{\geq 1}$ let
%% $n_i=k^i-k^{i-1}-k^{i-2}-\dots-1=\frac{(k-2)k^i+1}{k-1}$ and define the error
%% vector $v_i$ by the recurrence  
%% \[\left\{
%% \begin{array}{l}
%%   v_1=(\underbrace{0,\ldots,0}_{k-1 \text{ times}}) \in K^{n_1}\\
%%   v_{i} = (\underbrace{v_{i-1},\overbrace{1,\ldots,1}^{k^{i-1}
%%       \text{ times}},\ldots,v_{i-1},\overbrace{1,\ldots,1}^{k^{i-1} \text{ times}}}_{k-2
%%     \text{ times}},v_{i-1} ) \in K^{n_{i}}
%% \end{array}\right.
%% \]

%% There is $Z_i=(k-1)^{i}$ zeros in $v_i$ and $E_i=n_i-(k-1)^i$ ones.

\begin{lemma}\label{lem:lowerbound}
 Let $S$ be the support of $v_i$.
 If $k$ is prime, there is no $r\in \Z_{\geq 0}, s\in \Z_{>0}$  with  $r+s(k-1)<n_i$ such that 
$ \{r, r+s,r+2s, \dots r+(k-1)s\} \cap S = \emptyset.$
\end{lemma}

\begin{proof}
  Let $r \in \Z_{\geq 0}, s\in \Z_{>0}$ such that $r+(k-1)s<n_i$.
% and define  $T=\{r,r+s,\ldots,r+(k-1)s\}$.
Let $\ell$ be the multiplicity of $k$ in $s$ (possibly zero) and define $\alpha$
and $\beta$ such that $s=\alpha k^\ell + \beta k^{\ell+1}$ with $1\leq \alpha <k$.
%% Let $s_0=s$ and $\ell_0=\lfloor \log_k s_0\rfloor$ satisfying $k^{\ell_0}\leq s<k^{\ell_0+1}$. Note
%% that $\ell_0<i$. 
%% Consider the Euclidean division of $s_0$ by $k^\ell_0$: $s_0=\alpha_0
%% k^\ell_0+s_1$ where $0\leq s_1<k^{\ell_0}$ and $0 \leq \alpha_0 \leq  k-1\}$.
%% Proceed iteratively on $s_1,s_2 \ldots$ until $s_p=0$.

%% We have $s = \alpha_0k^{\ell_0} + \alpha_1k^{\ell_1}+\ldots +\alpha_pk^{\ell_p}$
%% with $0\leq \alpha_1,\ldots,\alpha_p <k$ and $\ell_0>\ell_1>\ldots>\ell_p$.

Let $\overline{r},\mu,\nu$ be such that $r=\overline{r}+\mu k^\ell+\nu
k^{\ell+1}$ with $0\leq \overline{r}<k^\ell$ and $0\leq \mu<k$.
As $\gcd(\alpha,k)=1$ there exists $1\leq j<k$ such that $j\alpha =k-1-\mu
\mod k$. Hence $j\alpha+\mu = k-1 + \lambda k$ for some $\lambda\in \Z$.
As $j<k$, the set $P_{r,s}$ contains the element $x=r+js$ and we write
\begin{eqnarray*}
x&=&r+js = \overline{r}+(j\alpha+\mu)k^{\ell} + \nu k^{\ell+1}\\
& =&\overline{r}+(k-1)k^{\ell} + (\nu+\lambda) k^{\ell+1}.
\end{eqnarray*}

We now show that the element of index $x$ in $v_i$ is a one. 
In this last expression, the term $(\nu+\lambda)k^{\ell+1}$ indicates that $x$
is located in the $(\nu+\lambda+1)$-st block of the form $(v_{\ell}, \overbrace{1,\ldots,1}^{k^\ell
  \text{ times}})$. Then the term $(k-1)k^\ell=m_\ell$ is precisely the
  dimension of $v_\ell$. Lastly, as
$\overline{r}<k^\ell$, we deduce that the element of index $x$ is a $1$ in~$v_i$.
%
%% Let $\ell=\lfloor \log_k s\rfloor$ satisfying $k^{\ell}\leq s<k^{\ell+1}$. Note
%% that $\ell<i$. Consider the Euclidean division of $s$ by $k^\ell$: $s=\alpha
%% k^\ell+t$ where $0\leq t<k^\ell$ and $0 \leq \alpha \leq  k-1\}$.
%% Here we distinguish two cases:
%% \begin{enumerate}
%% \item If $t=0$. Then $\alpha\geq 1$.
%% Let $\overline{r},\beta,\gamma$ be such that $r=\overline{r}+\beta k^\ell+\gamma
%% k^{\ell+1}$ with $0\leq \overline{r}<k^\ell$ and $0\leq \beta<k$.
%% As $\gcd(\alpha,k)=1$ there exist $1\leq j<k$ such that $j\alpha =k-1-\beta
%% \mod k$. Hence $j\alpha+\beta = k-1 + \delta k$.
%% As $j<k$, the set $T$ contains the element $x=r+js$ and we write
%% \[
%% r+js = \overline{r}+(j\alpha+\beta)k^{\ell} + \gamma k^{\ell+1} =
%% \overline{r}+(k-1)k^{\ell} + (\gamma+\delta) k^{\ell+1}.
%% \]
%
%% We now show that the element of index $x$ in $v_i$ is a one. 
%% In this last expression, the term $(\gamma+\delta)k^{\ell+1}$ indicates that
%% $\gamma+\delta$ blocks of the form $(v_{\ell}, \overbrace{1,\ldots,1}^{k^\ell
%%   \text{ times}})$ need to be skipped. Then the term $(k-1)k^\ell=n_\ell$
%% indicates that a block $v_\ell$ need to be skipped. Lastly, as
%% $\overline{r}<k^\ell$, we know that this index points to a one in $v_i$.
%% \item  {If $t\neq 0$.}
%% \end{enumerate}
\end{proof}

%% \begin{lemma}
%% If $k$ is prime, then $n_{k,E_i}= n_i+1$.
%% \end{lemma}

%% \begin{proof}
%% Lemma~\ref{lem:lowerbound} states that $n_{k,E_i}\geq n_i+1$.
%% There remain  to show that $n_{k,E_i}\leq n_i+1$, or
%%   equivalently that there is no subset of $E_i$ elements of $\{0,\ldots, n_i\}$
%%   that intersects all affine sub-sequences of length $k$.
%% We will prove it by strong induction on $i$. Consider the induction hypothesis:
%%   \begin{equation}
%%   H_j: %\forall j\leq i 
%% \left\{
%%   \begin{array}{l}
%%       n_{k,E_j} = n_j+1 \text{ and}\\
%% \forall S \subset \{0,\ldots,n_j-1\} \text{
%%  with $|S|=E_j$} \text{ not}\\
%% \text{verifying~(\ref{eq:affineseq}), }
%%  S\cap \{0,k^j, 2k^j,\ldots, (k-2)k^j\}=\emptyset.
%%   \end{array}\right.
%% \end{equation}

%% For $i=0, E_0=0$, the solution $n_{k,0}=k$ is obvious.

%% For a given $i$, suppose $H_i$ is true. Suppose $n_{k,E_{i+1}}>n_{i+1}+1$. Then
%% there exists a subset $S\subset\{0,\ldots,n_{i+1}\}$ of $E_{i+1}$ elements, not
%% verifying~(\ref{eq:affineseq}). 

%% Then $\forall p\in\{0,\dots,k-2\} \ |S\cap \{pk^i,\dots,(p+1)k^i-1\}|=E_i$, otherwise
%% $H_i$ would not be satisfied. 
%% This means that all elements in $S$ are smaller than $n_{i+1}$.
%% Moreover, $S\cap\{0,k^i,\ldots,(k-2)k^i\}=\emptyset$ by induction hypothesis.
%% Finally, the affine  sub-sequence of length $k$ $(0,k^i,2k^i,
%% \ldots, (k-1)k^i)$ does not intersect $S$ which is absurd.
%% \end{proof}

\begin{remark}
As suggested by a referee, we remark that problem~\ref{pb:affineseq} is closely
related to the famous problem of finding the largest sub-sequence of
$\{1,\dots,n\}$ not containing $k$ terms in arithmetic progression. 
Let $r(k,n)$ denoted the size of such a largest sub-sequence.
If $n\geq r(k,n)+E+1$, a subset of $E$ errors can not suffice to
intersect all arithmetic progressions of $k$ terms. Hence $n_{k,E}=\min\{n:
n-r(k,n)\geq E+1\}$. Noting that $r(k,n)\leq r(k,n+1)\leq r(k,n)+1$, we deduce
that for given $k$ and $E$, there always exists a $n^*$ such that $n^*-r(k,n^*)=E+1$ and consequently
$n_{k,E}=n^*=r(k,n^*)+E+1$. 
The value $r(k,n)$ has been first studied by Erd\H{o}s %ELK3: Hungarian umlaut
and Tur\'an~\cite{ErTu36}
who conjectured that for all $ k\geq 3$, $\lim_{n\rightarrow \infty} r(k,n)/n = 0$ which was proven by Szemeredi~\cite{Sze75}. 
In particular the construction of a bad error vector $w_i$ for $k$ prime has
connections with a construction of~\cite{ErTu36,Wag72}: its support is
formed by any element of $\{1,\dots,n\}$ whose base $k$ expansion contains at
least one digit equal to  $k-1$. 
This yields to the estimate 
\begin{equation}\label{eq:szerekes}
r\left(k,\frac{(k-2)k^i+1}{k-1}\right)\geq (k-1)^i.
\end{equation}
Szerekes conjectured that equality held in~(\ref{eq:szerekes})
(see~\cite{ErTu36}) which was disproved by Salem and Spencer~\cite{SaSp50}. 
\end{remark}
%CP Revision Ref3
%\begin{cor}\label{cor:nkE}
%If $k$ is prime, then $n_{k,E_i} = n_i+1$.
%\end{cor}

The error correction rate of the affine sub-sequence list decoding is
therefore directly related to the growth of the ratio $r(k,n)/n$ which is a core
problem in additive combinatorics. 
\begin{equation}
\frac{E}{n} = 1-\frac{r(k,n)}{n}-\frac{1}{n}.
\end{equation}
Szemeredi's theorem states that arithmetic progressions are dense,
i.e. an asymptotically large number of errors is necessary to intersect all of
them and rule out any list decoding possibility. Now there is unfortunately no
known expression of  $r(k,n)/n$ as a function of the information rate $k/n$, to
the best of our knowledge and we will now try to estimate bounds on this
decoding capacity. 

The error vectors $w_i$ approach a worst error distribution (but the result of Salem and
Spencer proves that it is not the worst case one). Consequently
we can derive from equation~(\ref{eq:szerekes})  an upper bound on the
maximal correction radius $E$: for $n=\frac{(k-2)k^i+1}{k-1}$ we have

\begin{eqnarray*}
E=n-r(k,n)-1%&\leq& \frac{k-2}{k-1}k^i+\frac{1}{k-1}-(k-1)^i\\
&\leq&k^i-(k-1)^i-\frac{k^i-1}{k-1}-1\\
&\leq&(i-1)k^{i-1} - \frac{k^{i-1}-1}{k-1}-1,
\end{eqnarray*}
as the function $f(x)=x^i$ is convex. Hence
$$
E\leq (i-1)k^{i-1}-\frac{k^{i-1}}{k-1}-\frac{k-2}{k-1}
$$
As $k^{i-1}=\frac{(k-1)n-1}{k(k-2)}$ we have $i-1\leq \log_k \frac{n}{k-2}$ %and $i-1\leq \log_k \frac{n}{k-2}$ 
and $\frac{k^{i-1}}{k-1}=\frac{n}{k(k-2)}-\frac{1}{k(k-2)(k-1)}$, therefore
\begin{eqnarray}\label{eq:upperbound}
E&\leq& \frac{n}{k-2} \left(\log_k \frac{n}{k-2}-\frac{1}{k}\right)-\frac{k-2}{k}\\
&\leq& \frac{n}{k-2} \log_k \frac{n}{k-2} \notag
\end{eqnarray}

This shows that, in the worst case, the improvement of the  affine sub-sequence technique
to  the correction radius, compared to the previous list decoding
($\frac{n}{k}-1$) is essentially no bigger than a logarithmic factor. 

The task of bounding $E$ or equivalently $E/n$ from below is much harder. In~\cite{Ro53},
Roth proved $r(3,n)\leq\frac{c}{\log \log n}$, leading to
$\frac{E}{n}\geq 1-\frac{c}{\log\log n}$  but for an arbitrary
$k$ the best known bound is given by Gowers~\cite{Go01}:
$$
\frac{E}{n} \geq 1-\frac{1}{(\log \log n)^{1/2^{2^{k+9}}}}
$$

Figure~\ref{fig:uplo} compares the upper bound on the correction capacity  $E$ of equation~(\ref{eq:upperbound}) with the actual values of Table~\ref{tab:bestn} for $k=5,7$.
%% The worst case upper bound states that in some cases, the improvement in the
%% decoding radius will never exceed a logarithmic factor  in $\frac{n}{k}$. 
%In Figure~\ref{fig:uplo}, these points below this bound can be seen  for $k=3$, $E=2,10$ or $k=5$, $E=4$.
\begin{figure}[h!]
\begin{center}
\includegraphics[width=\graphl]{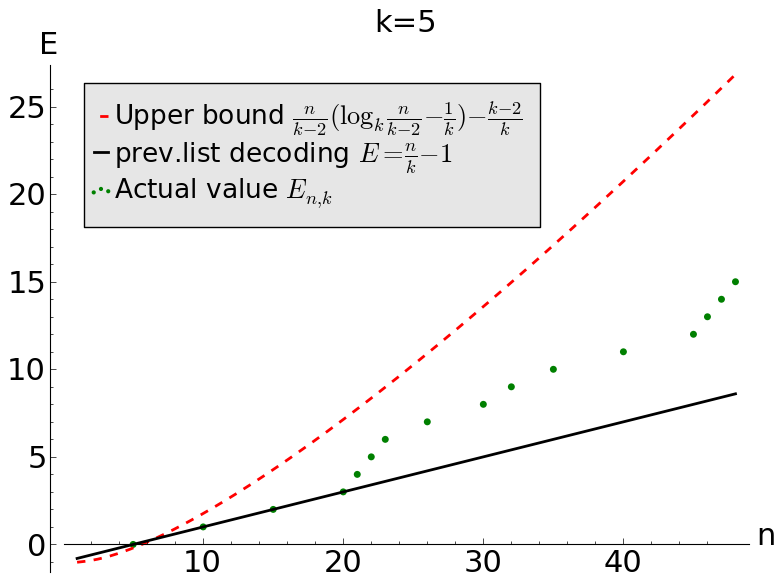}
\includegraphics[width=\graphl]{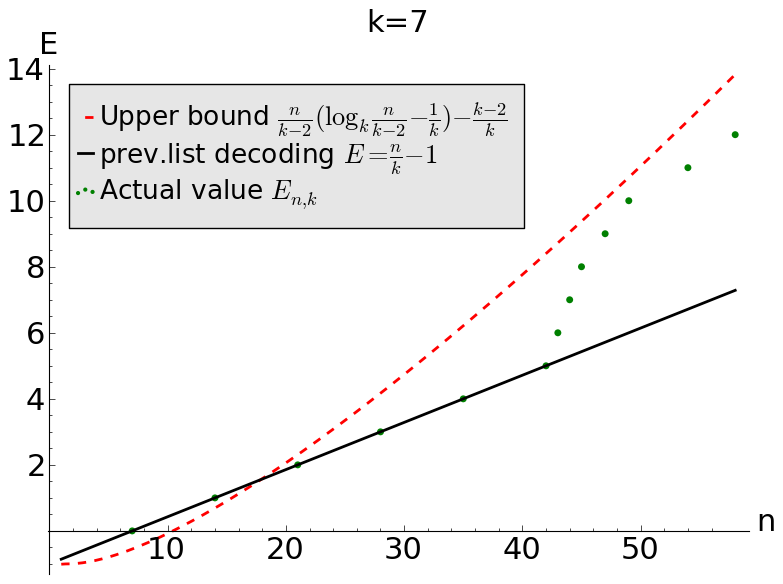}
\end{center}
\caption{Decoding radius for the affine sub-sequence algorithm: comparison of
upper, lower bounds and actual value}
\label{fig:uplo}
\end{figure}

\section{Characteristic zero}
\label{sec:charzero}

In this last section we  consider the case where  the base field has
 characteristic zero. We show  that some choices of evaluation points allow to
 reach much better minimum distances for sparse polynomial 
 evaluation codes.
\newline

\myparagraph{Positive real evaluation points}

\begin{theorem}\label{th:realmindist}%EK2 add
Consider $n$ distinct positive real numbers $\xi_0,\ldots,\xi_{n-1} >0$.
The sparse polynomial evaluation code defined  by
\[
\begin{split}
\mathcal{C}(n,\spbound) = \left\{ (f(\xi_0), \dots,f(\xi_{n-1})) :
    f\in \mathbb{R}[z] \text{ is }t\text{-sparse }\right.\\ 
\left.  \text{with } t \leq \spbound \right\}
\end{split}
\]
has minimum distance $\delta = n-2\spbound+1$.
\end{theorem}

\begin{proof}
Consider the code words $(f(\xi_0), \dots, f(\xi_{n-1}))$
and $(g(\xi_0), \dots,g(\xi_{n-1}))$ for a $t_f$-sparse polynomial $f$ and a
$t_g$-sparse polynomial $g$, with $t_f,t_g\le T$, at Hamming distance $\le n-2\spbound$.

Then the polynomial $f-g$ has sparsity $\le 2\spbound$, and vanishes 
in least $2\spbound$ distinct positive reals $\xi_i$.
By Descartes's rule of sign $f-g = 0$. 
\end{proof}
\begin{cor}
Suppose we have, for a $t_f \le \spbound$ sparse real polynomial $f(x)$, values
$f(\xi_i)+\epsilon_i$ for $2\spbound + 2E$ distinct positive real numbers $\xi_i
> 0$, where $e \le E$ of those values can be erroneous: $\epsilon_i\neq 0$.
If a $t_g \le \spbound$ sparse real polynomial $g$ interpolates any $2\spbound+E$ of
the $f(\xi_i)+\epsilon_i$, then $g = f$.
\end{cor}

So $f$ can be uniquely recovered from $2\spbound+2E$ values with $e\le E$ errors.
\begin{remark}
We do not have %EK2 provide
an efficient (in polynomial time) decoder up to half this minimum
distance. However, notice %EK2 remark
that when choosing the evaluation points $\xi_i
= \alpha^i$ for some $\alpha \in \RR_{>0} \backslash \{1\}$, the list decoder presented in
Section~\ref{sec:affineseq} %EK2 caps, as in Mister Pernet, not mister Pernet
can be used. Interestingly, it turns out to be  a
unique decoder as long as an affine sub-sequence free of error exists. 
Indeed, the list of candidates can be sieved by removing the polynomials which
evaluations differ by more than $\delta/2$ positions with the received word.
Finally the minimum distance of Theorem~\ref{th:realmindist} %EK2 caps and real!
ensures that only one code-word lies within less than $\delta/2$ modifications
of the received word, hence the decoding is unique.
\end{remark}

\myparagraph{Sampling primitive elements of co-prime orders in the complex unit
 circle} 
\begin{theorem}\label{th:coprimeorders}
Let $T, D$, and  $n\ge k = 2T \frac{\log(D)}{\log(2T)}$ be given.
Consider $n$ $p_i$-th roots of unity $\xi_i \ne 1$, where 
$2T < p_0 < p_2 < \cdots< p_{n-1}$, $p_i$ prime.
The sparse polynomial evaluation code defined  by
\[
\begin{split}
\mathcal{C}(n,\spbound) = \left\{ (f(\xi_0), \dots,f(\xi_{n-1})) :
    f\in \mathbb{Q}[z] \text{ is }t\text{-sparse }\right.\\ 
\left.  \text{with } t \leq \spbound \right\}
\end{split}
\]
has minimum distance $$\delta = n-k+1 = n-2\spbound\frac{\log(D)}{\log(2\spbound)}+1.$$
\end{theorem}

\begin{proof}
Consider the code-words
$(f(\xi_0), \dots, f(\xi_{n-1}))$
and $(g(\xi_0), \dots,g(\xi_{n-1}))$ for a $t_f$-sparse polynomial $f$ and a
$t_g$-sparse polynomial $g$, with $t_f,t_g\le T$,  at Hamming distance $\le n-k$.
Then $(f-g)(\zeta_j)$ vanishes for at least $k$ of the $\zeta_i$, say
for those sub-scripted $j \in J$.
%  The argument that then $f-g = 0$ is different. 

Let $0 \le e_1 < e_2 < \cdots < e_t$ be the term exponents in $f-g$,
with $t \le 2T$.  Suppose $f-g \ne 0$.  Consider
$
  M = (e_t - e_1) (e_t - e_2)\cdots (e_t-e_{t-1}).
  $
  Since $M \le D^{2T}$ and $\prod_{j\in J} p_j > (2T)^k \ge D^{2T}$, not all $p_j$
  for $j \in J$ can divide $M$.  Let $\ell \in J$ with $M \not\equiv 0 \pmod{p_\ell}$.
  Then the term $x^{e_t \bmod{p_\ell}}$ is isolated in
  $h(x) = (f(x)-g(x) \bmod{(x^{p_\ell} - 1)})$, and therefore the polynomial $h(x)$
  is not zero;  $h$ has at most $2T$ terms, and $h(\zeta_\ell) = 0$.
  This means that $h(x)$ and $\Psi_\ell(x) = 1 + x + \cdots + x^{p_\ell-1}$ have
  a common GCD.  Because $\Psi_\ell$ is irreducible over $\QQ$, and since
  $\deg(h) \le p_\ell-1$, that GCD is $\Psi_\ell$.  So $h$ is a scalar multiple
  of $\Psi_\ell$ and has $p_\ell > 2T$ non-zero terms, a contradiction. 
\end{proof}

\begin{cor}
Let $T, D, E$ be given and let the integer $k \ge 2T \log(D)/\log(2T)$.
Suppose we have, for a $t_f$-sparse polynomial $f\in\QQ[x]$, where
$t_f \le T$ and $\deg(f) \le D$, the values $f(\zeta_i)$ for $k + 2E$
$p_i$-th roots of unity $\zeta_i \ne 1$, where $2T < p_1 < p_2 < \cdots
< p_{N + 2E}$, $p_i$ prime.  Again $e \le E$ of those values can be
erroneous $f(\zeta_i) + \epsilon_i$.  If a $t_g$-sparse polynomial $g\in\QQ[x]$
with $t_g \le T$ and $\deg(g) \le D$ interpolates any $k+E$ of the
$f(\zeta_i)+\epsilon_i$, then $g = f$.
\end{cor}

\input{future}

\section{Acknowledgments}
We are thankful to  Daniel Augot, Bruno Salvy and the  referees for their helpful
remarks and suggestions.
%% \def\refname{\Large\bfseries References}

%% %mtc3 \bibliographystyle{myplainnat}
%% \bibliographystyle{acm} %ek5 {abbrv}

%% \bibliography{strings,ldsic,crossrefs}

%% \end{document}

%% file: abstract.tex
We present algorithms performing sparse univariate polynomial interpolation with
errors in the evaluations of the polynomial. Based on the initial work by Comer,
Kaltofen and Pernet [Proc. ISSAC 2012], we define the sparse polynomial
interpolation codes and state that their minimal distance is precisely the code-word length divided by twice the sparsity. At ISSAC 2012, we have given a decoding algorithm for as much as half the minimal distance and a list decoding algorithm up to the minimal distance. 

Our new polynomial-time list decoding algorithm uses sub-sequences of the
received evaluations indexed by an arithmetic progression, allowing the decoding for a larger radius, that is, more errors in the evaluations while returning a list of candidate sparse polynomials. We quantify this improvement for all typically small values of number of terms and number of errors, and provide a worst case asymptotic analysis of this improvement. For instance, for sparsity $T = 5$ with $\le 10$ errors we can list decode in polynomial-time from $74$ values of the polynomial with unknown terms, whereas our earlier algorithm required $2T(E+1) = 110$ evaluations. 

We then propose two variations of these codes in characteristic zero, where
 appropriate choices of values for the variable yield a much larger minimal
 distance: the code-word length minus twice the sparsity.
% For those evaluations, our list decoding algorithm yields a single unique sparse interpolant. 

%% file: 1intro.tex
\section{Introduction}\label{sec:intro}
Evaluation-interpolation schemes are %ELK3 was: is
a key ingredient in many of today's %ELK3 was: todays
computations. 
Model fitting for empirical data sets is a well-known one, where additional
information on the model helps improving the fit. In particular, models of
natural phenomena often happen to be sparse, which has motivated a wide range of
research including compressive sensing~\cite{CT04}, and sparse
interpolation of polynomials~\cite{deProny,B-OTi88,KLW90,GriKar93,Garg20092659,GR10}.
Most algorithms for the latter problem rely on the connection between 
linear complexity and sparsity, often  referred to as Blahut's Theorem
(Theorem~\ref{th:Blahut} \cite{Blahut83,MasseySchaub88}) 
though
%although it was 
already used in the 18th
century by Prony~\cite{deProny}. The \BMA~\cite{Mass69} makes this connection
effective. These exact sparse interpolation techniques have been very
successfully applied to numeric computations~\cite{GLL09,KL03,CKP12,KaYa13}.
 
Computer algebra also widely uses evaluation\--inter\-pola\-tion schemes as a
key computational tool: reducing operations on polynomials to base ring 
operations, integer and rationals operations to finite fields
operations, multivariate polynomials operations to univariate polynomials
operations, etc. With the rise of large scale parallel computers, their ability to convert
a large sequential computation, into numerous smaller independent tasks is
of high importance.

%\myparagraph{Error correcting codes}

Evaluation-interpolation schemes are also at the core of the famous Reed-Solomon
error correcting codes~\cite{ReSo60,Moon05}. There, a block of information, viewed as a
dense polynomial over a finite field is encoded by its evaluation in $n$
points. Decoding is achieved by an interpolation resilient to errors.
%
%
%% The source message is a polynomial of
%% degree $<k$, and its encoding is a vectors of $n\geq k$ of its evaluations. The
%% oversampling by $n-k$ points is a redundancy that makes it possible to correct
%% errors in some of the evaluations. The minimal distance between all non-zero
%% code words
%% characterize the maximal number of errors that can be corrected. These codes
%% have minimal distance $n-k$ and can therefore correct up to
%% $\lfloor\frac{n-k)}{2}\rfloor$ errors.
%
Blahut's theorem~\cite{Blahut83,MasseySchaub88} originates from the decoding of Reed-Solomon codes: the
interpolation of the error vector of sparsity $t$ is a sequence of linear
complexity $t$ whose  generator, computed by \BMA, carries  in its roots the
information of the error locations.
Beyond the field of digital communication and data storage, error correcting
codes have found more recent applications in fault tolerant distributed
computations~\cite{KH84,KPRRS10,DBBHD12}. In particular, %ELK3 omit: a
parallelization based
on evaluation-interpolation %ELK3 omit techniques
can be made fault tolerant if
interpolation with errors %ELK3 plural
is performed. This is achieved by Reed-Solomon
codes for dense polynomial interpolation and by CRT codes, for residue number
systems~\cite{KPRRS10}. 
The problem of sparse polynomial interpolation with errors rises
naturally in this context.
%ELK3 rewrite: We introduced it in~\cite{CKP12} and proposed algorithms to solve it. 
We give algorithms for the solution of the problem in~\cite{CKP12}.
Our approach is %ELK3 was: It is
naturally related to the $k$-error linear complexity problem~\cite{MeNi02}
from stream cipher theory. 
A major concern in our previous results is that in
order to correct $E$ errors, the number of evaluations has to be increased by a
multiplicative factor linear in $E$. In comparison, dense interpolation with
errors only requires an additive term linear in $E$. 

In this paper we further investigate this problem from a coding theory
viewpoint. 
In section~\ref{sec:sic} we define the sparse polynomial
interpolation codes. We then focus on the case where the evaluation points are
consecutive powers of a primitive root of unity, whose %ELK3 was: which
order is divisible by twice
the sparsity, in order to %ELK3 was: so as
to benefit from Blahut/Ben-Or/Tiwari
%ELK3 was: Ben-Or/Tiwari's %CP2: there is no Blahut's algorithm. I prefer to
%revert to Ben-Or/Tiwari  alg 
interpolation algorithm. We
show that in this stetting %ELK3 was: context
the minimal distance is precisely the
length divided by twice the sparsity. The algorithms of~\cite{CKP12} can be viewed
as a unique decoding algorithm for as much as half the minimal distance and a
list decoding algorithm up to the minimal distance.
In section~\ref{sec:affineseq}, we propose a new polynomial-time list decoding
algorithm that uses sub-sequences of the received evaluations indexed by an
arithmetic progression, reaching a larger decoding radius.
%, that is, more errors in the evaluations while returning a list of candidate
%sparse polynomials.
We quantify this improvement on average by experiments, in the worst case for
all typically small values of number of terms and number of errors, and make
connections between the asymptotic decoding capacity and %ELK3 remvoe a
the famous 
Erd\H{o}s-Tur\'an %ELK3 add
problem of additive combinatorics. %ELK3 plural!
%% For instance, for sparsity $T = 5$ with $\le 10$
%% errors we can list decode in polynomial-time from $74$ values of the polynomial
%% with unknown terms, whereas our earlier algorithm required $2T(E+1) = 110$
%% evaluations.  
%
We then propose in section~\ref{sec:charzero} two variations of these codes in
characteristic zero, where appropriate choices of values for the variable yield
a much larger minimal distance: the length minus twice the sparsity. 
%For those
%evaluations, our list decoding algorithm yields a single unique sparse interpolant.  

\myparagraph{Linear recurring sequences} 
We recall that a sequence $(a_0,a_1,\ldots)$ is linearly recurring if there exists
$\lambda_0, \lambda_1,\allowbreak \dots, \lambda_{t-1}$ such that $a_{j+t} =
\sum_{i=0}^{t-1}\lambda_ia_{j+i} \ \forall j\geq 0$. 
The monic polynomial $\Lambda(z) = z^t-\sum_{i=0}^{t-1}\lambda_i z^i$ is called a
generating polynomial of the sequence, the  generating polynomial with least
degree is called the minimal generating polynomial and its degree is the linear
complexity of the sequence.

These definitions can be extended to vectors, viewed as contiguous
sub-sequences of an infinite sequence. The minimal generating polynomial of an
$n$-dimensional vector is the monic polynomial $\Lambda(z) =
z^t-\sum_{i=0}^{t-1}\lambda_i z^i$ of least degree such that 
$a_{j+t} = \sum_{i=0}^{t-1}\lambda_ia_{i+j}\  \forall 0\leq j\leq n-t-1$.
Note that consequently, any vector is linearly recurring with linear
complexity less than $n$.
\begin{theorem}  [Blahut~\cite{Blahut83, MasseySchaub88}]\label{th:Blahut}
Let $\KK$ be a field containing an $N$-th primitive root of unity.
The linear complexity of an $N$-periodic sequence $A=(a_0, \dots, \allowbreak
a_{N-1},\allowbreak a_0,\ldots)$ over $\KK$ is equal
to the Hamming weight of the discrete Fourier transform of
$(a_0,\ldots,\allowbreak a_{N-1})$.
\end{theorem}

\myparagraph{The Blahut/Ben-Or/Tiwari Algorithm}
We  review the Blahut/Ben-Or/Tiwari~\cite{Blahut84,B-OTi88} algorithm in the setting of univariate sparse
polynomial interpolation. 
Let $f$ be a univariate polynomial with $\gendeg$ terms, $m_j$ and let $c_j$ the corresponding non-zero coefficients:
\newline\centerline{
$
\sparpoly(x)=\sum_{j=1}^%
\gendeg c_jx^{e_j}=\sum_{j=1}^%
\gendeg c_jm_j %
\neq 0, e_j\in \ZZ.
$
}%

%\vspace{-2ex}
\begin{theorem}%
\label{Alg:B-OTi}
{\upshape\citep{B-OTi88}} %
Let $b_j=%
\alpha^{e_j}$, where $\alpha$ is a value from the coefficient domain to be
specified later, let
$
a_i$ $=$ $f(%
\alpha^i)$ $=$ $\sum_{j=1}^%
\gendeg c_j b_j^i,%
$
and let
$\genf(z)=\prod_{j=1}^\gendeg (z-b_j) = z^\gendeg +\lambda_{\gendeg-1}z^{\gendeg-1}+\dots+\lambda_0$.
The sequence %
$\seq{a_0,a_1,\ldots}$ is linearly
generated by the minimal polynomial %
$\genf(z)$.
\end{theorem}
The Blahut/Ben-Or/Tiwari algorithm then proceeds in the four following steps:
\par\noindent\hangindent=1.3em\hangafter=1 %
1.
Find the minimal%
-degree generating polynomial %
$\genf$ for %
$\seq{a_0,a_1,\ldots}$, using %
the \BMA.
\par\noindent\hangindent=1.3em\hangafter=1
2.
Compute %
the roots $b_j$ of $\genf$,
using univariate polynomial factorization.
\par\noindent\hangindent=1.3em\hangafter=1
3.
Recover the exponents $e_j$ %
of $\sparpoly$, %
by repeatedly dividing $b_j$ by %
$\alpha$. %
\par\noindent\hangindent=1.3em\hangafter=1
4.
Recover the coefficients $c_j$ %
of $\sparpoly$, by solving the transposed %
$\gendeg\times\gendeg$
  Vandermonde system

\centerline{$
\begin{bmatrix}
1      & 1      & \dots  & 1 \\
b_1    & b_2    & \dots  & b_\gendeg \\ %
\vdots & \vdots & \ddots & \vdots \\
b_1^{\gendeg-1}  & b_2^{\gendeg-1}  & \dots  & b_\gendeg^{\gendeg-1} 
\end{bmatrix} 
\begin{bmatrix}
c_1\\c_2\\\vdots\\ c_\gendeg %
\end{bmatrix}
=
\begin{bmatrix}
a_0\\a_1\\\vdots\\ a_{\gendeg-1} %
\end{bmatrix}.
$}
\par
By Blahut's theorem, %~\citep{Blahut83,MasseySchaub88}, %
the sequence $\seq{a_i}_{i\geq0}$ has linear complexity %
$\gendeg$,
hence only %
$2\gendeg$ coefficients suffice for %
the \BMA to recover
the minimal polynomial %
$\genf$.
In the presence of errors %
in some of the evaluations, this fails.

%% file: future.tex
\section{Conclusion}

Our %ELK3 was: These
codes, arising from a natural construction, are surprisingly rich and
difficult to analyze.  On one hand,
it is natural to choose evaluation points as consecutive powers of
a primitive root of unity, in order to benefit from the efficient interpolation
algorithm of Blahut/Ben-Or/Tiwari, %ELK3 was: Ben-Or \& Tiwari,
but it is precisely this setting that implies existence of bad worst case
error vectors and hence reduces their minimum distance. Much better minimum
distances should be attained in the general case, as suggested by
Theorem~\ref{th:coprimeorders}, but then no efficient decoding algorithm is
available.
%ELK3 Hence we are still at the very beginning in the study of these codes
%ELK3 and the algorithmic of their decoding.
Those are apparently difficult problems left to be solved. %ELK3 rewrite
%% The surprisingly large  minimum distance found in characteristic zero motivates
%% the search for efficient decoding algorithms reaching this decoding radius. 
%% %
%% A next step is also to find an analogous of theorem~\ref{th:coprimeorders} over
%% finite fields, using sampling points in extensions defined by an irreducible
%% polynomial with high Hamming weight.

%% Studying the decoding capacity of interleaved sparse polynomial
%% interpolation codes is also of interest. It should improve the decoding radius
%% and has direct applications to the recovery of vectors of sparse polynomials.